\documentclass[12pt]{amsart}

\usepackage[leqno]{amsmath}
\usepackage{amssymb}
\usepackage[normalem]{ulem}
\usepackage[]{graphicx}
\usepackage[margin=3cm]{geometry}
\usepackage[]{color}

\newtheorem{thm}{Theorem}[section]
\newtheorem{lem}[thm]{Lemma}

\newtheorem{cor}[thm]{Corollary}
\theoremstyle{definition}

\newcommand{\R}{\mathbb R}

\newcommand{\T}{\mathbf T}
\newcommand{\A}{\mathcal A}
\newcommand{\U}{\mathcal U}

\title[Identifiability of tree-child phylogenetic networks]{Identifiability of tree-child phylogenetic networks under a
	probabilistic recombination-mutation model of evolution} 

\author{Andrew Francis and Vincent Moulton}
\address{Centre for Research in Mathematics, Western Sydney University, Sydney, Australia}
\email{a.francis@westernsydney.edu.au}
\address{School of Computing Sciences, University of East Anglia, Norwich, UK}
\email{v.moulton@uea.ac.uk}

\date{\today}

\begin{document}
\begin{abstract}
Phylogenetic networks are an extension of phylogenetic trees which are used to 
represent evolutionary histories in which reticulation events 
(such as recombination and hybridization) have occurred. A central question
for such networks is that of \emph{identifiability}, which essentially asks 
under what circumstances can we reliably 
identify the phylogenetic network that gave rise to the observed data?
Recently, identifiability results have appeared for networks
relative to a model of sequence evolution that generalizes the 
standard Markov models used for phylogenetic trees. However, these results 
are quite limited in terms of the complexity of the networks that are considered.
In this paper, by introducing an alternative probabilistic model for
evolution along a network that is based on some ground-breaking work
by Thatte for pedigrees, we are able to obtain an identifiability result for a much larger class of 
phylogenetic networks (essentially the class of so-called 
tree-child networks). To prove our main theorem, we derive some
new results for identifying tree-child networks combinatorially, and then
adapt some techniques developed by Thatte for pedigrees  
to show that our combinatorial results imply identifiability in the probabilistic setting. 
We hope that the introduction of our new model for networks
could lead to new approaches to reliably construct phylogenetic networks.  
\end{abstract}

\maketitle

\section{Introduction}

Recently, there has been growing interest in the construction of phylogenetic
networks in order to represent the evolutionary history of a given set 
of species or taxa \cite{bapteste2013networks}. 
Phylogenetic networks are a generalization of phylogenetic trees,
and they have the advantage of being able to represent evolutionary events such
as recombination and hybridization that is not possible within a single tree. Various
approaches have been developed for constructing networks 
\cite{gusfield2014recombinatorics,huson2010phylogenetic}, and more recently
the use of probabilistic approaches for this purpose has started to gain momentum.

One of the key issues that arises when applying probabilistic models in 
phylogenetics is that of identifiability: under what circumstances can we reliably 
identify the phylogenetic tree or network that gave rise to the observed data?  
Typically, as is the case in this paper, the observed data is a multiple alignment of 
sequences across a set of taxa, which correspond to the leaves of the tree or network.   
This identifiability problem has been extensively studied 
for phylogenetic trees where identifiability has been proven
for simple models  some time ago (see e.g. \cite{chang1996full,steel1998reconstructing} as well as  \cite{rhodes2012identifiability} for an overview of some more recent developments), but 
relatively little is known for more general networks.

Identifiability results for phylogenetic networks come with two riders: the model 
of evolution considered on the network; and the class of networks considered.  
Most studies use a model of evolution on a rooted binary phylogenetic 
network, in which characters evolve along arcs, copy themselves at tree vertices, and 
make a random choice at reticulation vertices \cite{nakhleh2010evolutionary}. 
Under this model of evolution, identifiability results are known for a limited set of families of networks.  
For instance~\cite{gross2017distinguishing} have shown that under such a model, 
networks with a single cycle of length greater than or equal to 4 are identifiable. 
Related network-based models consider evolution along the trees that are contained within a network
and take into account processes such as incomplete lineage
sorting \cite{yu2014maximum},  but identifiability of
these models is complicated by the fact that the trees displayed in a network do not 
necessarily identify the network.

In this paper we consider a different evolutionary model, which we adapt from 
the world of pedigree reconstruction~\cite{thatte2013reconstructing}.  In this model, which
we illustrate in Figure~\ref{f:markov},
first a tree is selected at random from the set of trees displayed  by a  
network, and a standard model of evolution 
on that tree (see e.g. \cite[Chapter 13]{felsenstein2004inferring})  is used
to generate character values at sites until the tree changes.  At each site and for each reticulation 
vertex there is a fixed (small) probability $p$ that the parent of the vertex will 
switch.  For networks whose displayed trees have leaf-set equal to that of the network, this 
generates a Markov process whose state space is the set of 
displayed trees of the network, and means that an alignment will have 
blocks of sites generated under a common tree, before a change produces 
another block of sites generated under a new tree.  Related approaches have been 
considered in the literature for constructing networks (known as 
ancestral recombination graphs) from an alignment of
recombining sequences \cite{song2005constructing}, and also for inferring break-points in 
such alignments using the so-called multiple changepoint model \cite{suchard2003inferring}.

\begin{figure}[ht]
	\includegraphics{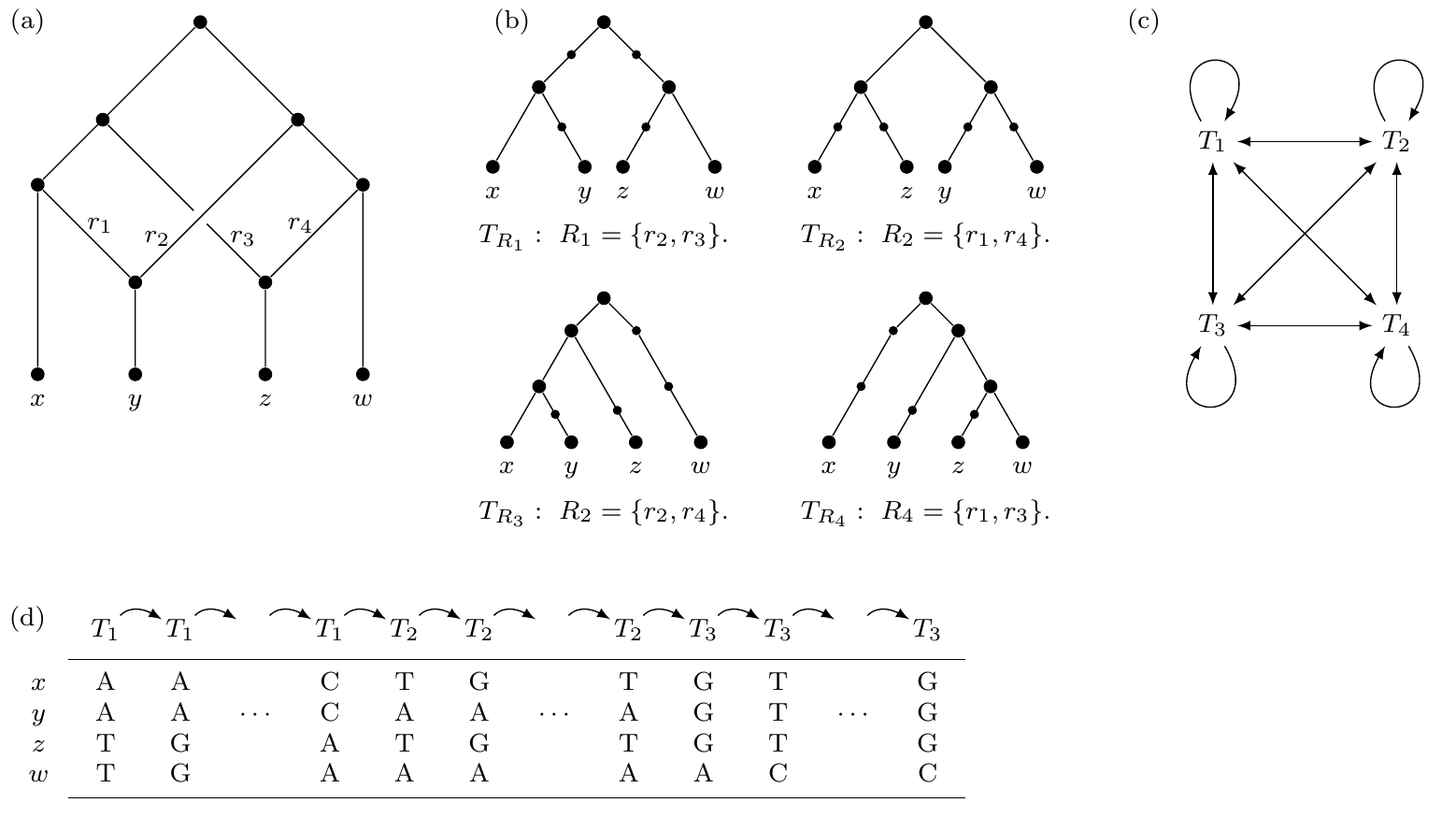}
	\caption{(a) A phylogenetic network with two reticulations and reticulation arcs labelled $r_i$ (Section~\ref{s:networks}).
		All arcs are directed downwards and away from the root. (b) The four rooted trees $T_{R_i}$ in the network corresponding to choices of sets of reticulation arcs $R_i$ (Section~\ref{s:tcn}). (c)~The Markov model on the four trees allows movement between any pair of trees by changing choice of $R_i$ (abbreviating $T_{R_i}$ to $T_i$). This is introduced in Section~\ref{s:evol.on.tree}. (d) An alignment generated by the 
	sequence of trees  $T_1\to T_1\to\dots\to T_1\to T_2\to T_2\to\dots\to T_2\to T_3\to T_3\to \dots \to T_3$ 
    (Section~\ref{s:evol.on.tree}).}\label{f:markov}
\end{figure}

Under our model of evolution along a phylogenetic network, we are able to obtain 
an identifiability result for a much larger class of networks than has been possible before.
In particular, in our main result, Corollary~\ref{second},  we show that
it is possible identify any network
within the class of tree-child networks, all of which have same number of reticulation vertices, 
and such that none of them has a reticulation vertex adjacent to the root (see Section~\ref{eq:AT} for
the definition of these terms).  Whereas for the model 
used in \cite{gross2017distinguishing} network identifiability has only been shown to hold 
for the case where there is a single reticulation, the networks that we consider 
can have any number of reticulations (if the number of leaves is allowed to grow).  

We now summarize the rest of the paper.
We begin with a section defining the key terms that will be used throughout 
the paper (Section~\ref{s:prelims}).   Section~\ref{s:tcn} provides key results on 
tree-child networks that we will need, some of them new.  In particular we prove 
that the number of non-isomorphic ``embedded spanning trees'' in a tree-child 
network is 2 raised to power of the number of reticulation vertices in the network (Theorem~\ref{tight}), and 
that if two tree-child networks have the same set of embedded spanning trees, then 
they are isomorphic (Theorem~\ref{isomorphic}).
In Section~\ref{s:models} we introduce the model of evolution that we will study on 
a network, based on that of~\cite{thatte2013reconstructing} for pedigrees, and 
we adapt the key results of~\cite{thatte2013reconstructing} for the setting of
rooted binary phylogenetic networks.  Our main result is contained in 
Section~\ref{s:main.results}, and it is based on a result which states that if the  distributions of the 
characters arising on certain networks are the same then, for sufficiently 
long alignments and a certain choice of model parameter, 
the networks must contain the same set of embedded spanning trees 
(Corollary~\ref{first}). This, in turn, is a direct corollary of  
a technical result (Theorem~\ref{t:main.result}) whose proof employs
a similar strategy to that used in the proof of  
\cite[Theorem 2]{thatte2013reconstructing}. We conclude with a short 
discussion on possible future directions.

\section{Preliminaries/Definitions}\label{s:prelims}

\subsection{Trees and forests}

In what follows $X$ is a finite set (corresponding to a set of taxa).

A \emph{forest} is a graph with no cycles; a \emph{tree} is a forest with one
connected component. A \emph{leaf} in a forest is a degree 1 vertex.
A \emph{rooted tree} is a tree 
with one vertex identified called the root, and all arcs directed away from the root towards the leaves. 
Note that if the root has out-degree 1, then we do not regard it as being a leaf of the tree.

Following \cite[Definitions 4 and 6]{thatte2013reconstructing}, we define 
an \emph{$X$-forest} to be a forest with leaf-set $X$, 
and say that two $X$-forests $F_1$ and $F_2$ are isomorphic if there is a graph 
isomorphism between $F_1$ and $F_2$ which is the identity on  $X$.
An {\em $X$-tree $T$}  is an $X$-forest 
with one component such that all internal vertices of $T$ have 
degree either 2 or 3.  Note, that $X$-forests are unrooted.
Moreover, it is important to note that an $X$-forest (or $X$-tree) may contain  
vertices of degree 2 that are not contained in $X$, and so the term is used in a 
slightly different way from that commonly used in the phylogenetics literature.

\subsection{Phylogenetic networks}\label{s:networks}

For networks we follow the definitions presented in \cite{bordewich2016determining}.

A \emph{phylogenetic network} $N$ on $X$ is a directed acyclic graph with the following properties:
(i) it has a unique vertex of in-degree zero called the \emph{root}, which has out-degree 
two (except in the case $|X| = 1$), (ii) the set $X$ labels the set of vertices of out-degree zero, 
each of which has in-degree one, and (iii) every other vertex either has in-degree one and out-degree 
two, or in-degree two and out-degree one. For technical reasons, in 
case $|X|=1$, then $N$ consists of the single vertex in X.

We denote the set of arcs in a phylogenetic network $N$ by $A(N)$.
The vertices of out-degree zero are called \emph{leaves}, while the vertices of in-degree 
one and out-degree two are called \emph{tree vertices} and the vertices of in-degree two and out-degree one 
are called \emph{reticulations}. The arcs directed into a reticulation are called 
\emph{reticulation arcs}; all other arcs are called \emph{tree arcs}.  
We let $r(N)$ denote the number of reticulations in $N$.
We say that two phylogenetic networks $N_1$ and $N_2$ are \emph{isomorphic} 
if there exists a directed graph isomorphism between $N_1$ and $N_2$
which is the identity when restricted to $X$.

For any two vertices $u$ and $v$ in $N$ that are joined by an arc $(u,v)$, we say 
$u$ is a \emph{parent} (or parent vertex) of $v$ and, conversely, $v$ is a \emph{child} (or child vertex ) of $u$. 
We say that $N$ is a \emph{tree-child network} if every non-leaf vertex has a child 
which is either a tree vertex or a leaf~\cite{cardona2009comparison}. 

\section{Tree child networks and embedded spanning trees}\label{s:tcn}

Given a phylogenetic network $N$, we can obtain a rooted tree from $N$ by removing one of the two 
reticulation arcs incident to each one of the reticulations in $N$. If $R$ denotes
a set of reticulation arcs removed in this way, then we let $T_R$ denote this tree.  
We let $\mathcal R_N$ denote the set of all possible such sets $R$ 
(so that $|\mathcal R_N|= 2^{r(N)}$).
Note that the vertex set of $T_R$ contains $X$ and it may potentially contain 
degree two vertices, as well as leaves that are not contained in $X$. 

Given a network $N$, we say that a tree whose vertex set contains $X$ is 
an \emph{embedded spanning tree} in $N$ if
it is isomorphic to the (unrooted) tree which is obtained from 
$T_R$ for some $R \in \mathcal{R}_N$, by ignoring directions on arcs,
via an isomorphism of trees which is the identity on $X$. 
We denote the set of all possible embedded spanning trees in $N$ (up to
isomorphism) by $S(N)$. Clearly $1 \le |S(N)| \le 2^{r(N)}$.  An example is shown in Figure~\ref{f:embedded.unrooted}.

\begin{figure}[ht]
\includegraphics{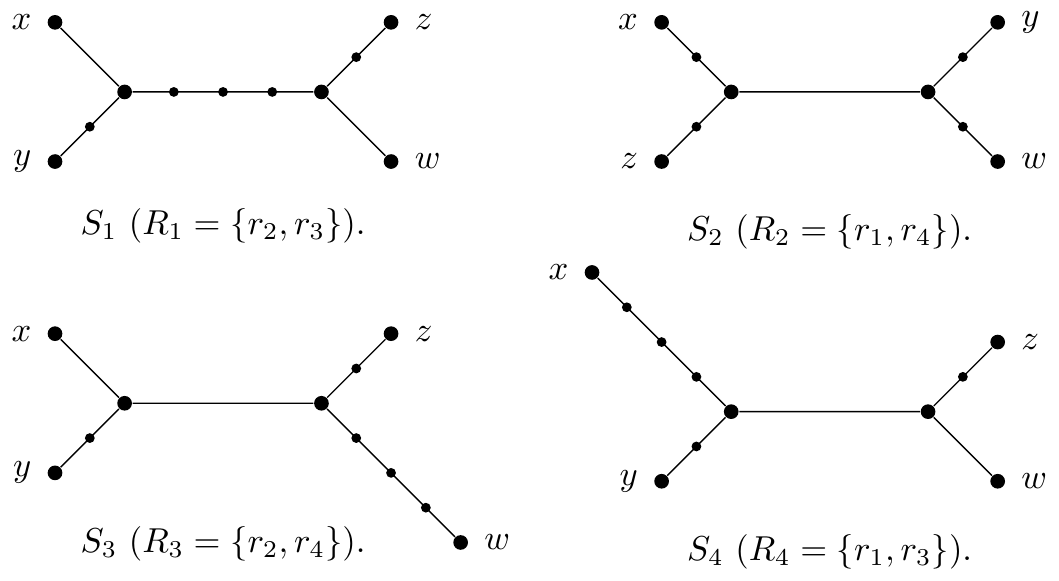}
\caption{The set of embedded spanning trees $S(N)=\{S_1,S_2,S_3,S_4\}$ for the network shown in Figure~\ref{f:markov}, and corresponding to the rooted analogues $T_{R_1},\dots,T_{R_4}$.}
\label{f:embedded.unrooted}
\end{figure}

Later, we shall focus on tree-child networks.
Note that an embedded spanning tree in a tree-child network $N$ on $X$
may not necessarily be an $X$-tree. We now characterize those
tree-child networks for which every embedded spanning tree is an $X$-tree.

\begin{lem}\label{special}
	Suppose that $N$ is a tree-child network. Then every element in $S(N)$ is 
	an $X$-tree if and only if there does not exist an arc $(\rho,v)$ in 
	$N$ with $\rho$ the root of $N$ and $v$ a reticulation of $N$.
\end{lem}
\begin{proof} 
Suppose that every element in $S(N)$ is 
an $X$-tree. Suppose $N$ contains an arc $(\rho,v)$ 
with $\rho$ the root of $N$ and $v$ a reticulation of $N$. 
If $R \in \mathcal{R}_N$ with $(\rho,v) \in R$,  
the underlying (undirected) tree of $T_R$ is a 
tree whose vertex set contains $X$ 
with a leaf (corresponding to $\rho$) that is not in $X$, a contradiction.

Conversely, suppose there does not exist an arc $(\rho,v)$ in 
$N$ with $\rho$ the root of $N$ and $v$ a reticulation of $N$.
Let $R \in \mathcal{R}_N$, and suppose that the 
embedded spanning tree arising from $T_R$ contains
a leaf $w$ that is not in $X$. 

Note first that $w$ is not the root of $N$, since otherwise there 
would be an arc $(w,v)$ in $N$ with $v$ a reticulation of $N$, which
is contradiction. So, as $w$ is not in $X$, it therefore follows that $w$ is either a reticulation 
or a tree-vertex. But $w$ cannot be a reticulation since
then there would have to be a reticulation $v$ with $(w,v)$ 
an arc in $N$, which contradicts $N$ being tree-child.  
Similarly, $w$ cannot be a tree-vertex, as then to have $w$ being a 
leaf in $T_R$, both of the children of $w$ in $N$ would
have to be reticulation vertices, which again contradicts $N$ being tree-child. 
This final contradiction completes the proof of the lemma.
\end{proof}

In the following 
we will use two operations on phylogenetic networks 
as defined in  \cite{bordewich2016determining}. 
Let $N$ be a phylogenetic network on $X$. 
A 2-element subset $\{x, y\}$ of $X$ 
is a \emph{cherry} in $N$ if the parents of $x$ and $y$ are the same.
A {\em cherry reduction} on a cherry $\{x, y\}$ in $N$ is the operation of deleting 
$x$ and $y$, and their incident arcs, and labelling their common parent 
(now itself a leaf) with a new element not in $X$. Note that
after a cherry reduction the number of leaves in the resulting 
network is reduced by one, but the number of reticulations is unchanged.

A two-element subset $\{x, y\}$ of $X$ is called a \emph{reticulated 
cherry} in $N$ if there is an undirected path, say $x, v_1, v_2, y$, 
between $x$ and $y$ in $N$ with one of $v_1$ and $v_2$ a tree vertex and 
the other a reticulation vertex. A {\em reticulated cherry reduction} 
on a reticulated cherry $\{x, y\}$ in $N$
is the operation of deleting the reticulation arc of the reticulated cherry 
and suppressing the degree-two vertices resulting from the deletion. 
Note that after a reticulated cherry reduction, the number of reticulations in the resulting 
network is reduced by one, but the leaf set is unchanged. 

The following result, that will be key to us, is shown in \cite{bordewich2016determining}. 

\begin{thm}[\cite{bordewich2016determining}]\label{t:cherries}
If $N$ is tree-child network on $X$, then the following hold:
\begin{enumerate}
	\item[(i)] If $|X| \ge 2$, then $N$ contains either a cherry or a reticulated cherry.
	\item[(ii)] If $N'$ is obtained from $N$ by reducing either a cherry or a reticulated
cherry, then $N'$ is a tree-child network.
\end{enumerate}
\end{thm}

Note that using this theorem it is straight-forward to check that in case 
$|X|=2$, then if $X=\{x,y\}$ a tree-child network on $X$ must be isomorphic to
one of the two networks in Figure~\ref{f:tcn.2leaves}.
 
\begin{figure}[ht]
\includegraphics[]{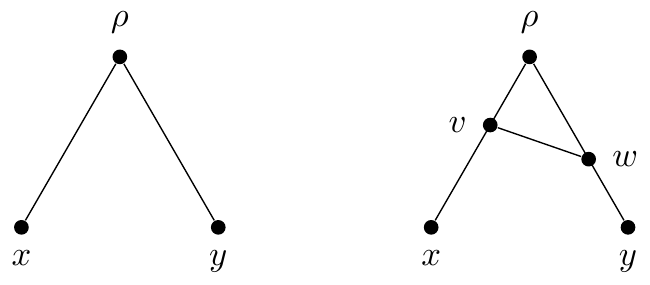}
\caption{The two tree-child networks with two leaves.}
\label{f:tcn.2leaves}
\end{figure}

We now prove that if $N$ is tree-child then $S(N)$ must be as large as is
possible for a network.

\begin{thm}\label{tight}
	Suppose that $N$ is a tree-child network. Then $|S(N)| = 2^{r(N)}$.
\end{thm}
\begin{proof}
We will show that if $R \neq R' \in \mathcal{R}_N$, then
the embedded spanning trees in $S(N)$ arising for $T_R$ and $T_{R'}$ (which we
denote by $S_R$ and $S_{R'}$, respectively) are not
isomorphic via an isomorphism of trees which is the identity on $X$. 

Suppose this is not the case. Let $X$ be of minimal size such that there 
is a tree-child network $N$ on $X$ with $R \neq R' \in \mathcal{R}_N$, but 
$S_R$ is isomorphic to $S_{R'}$. Moreover, out of all such
networks on $X$, pick $N$ which has a minimal number of arcs.
It is straight-forward to check using the above observation for 
tree-child networks with two leaves that $|X|$ must be greater than $2$.

Since $N$ is tree-child, it must contain a cherry or reticulate cherry (Theorem~\ref{t:cherries}(i)). If
it contains a cherry, then perform a cherry reduction on $N$ to obtain 
a tree-child network $M$. This reduction does not affect any 
reticulation arcs, and so $R$ and $R'$ are both subsets
of the arcs of $M$. Moreover, as $S_R$ is isomorphic to $S_{R'}$, 
this also holds for the reduced versions of $S_R$ and $S_{R'}$.
But this contradicts the fact that $X$ was chosen to be minimal, since
$M$ has a smaller leaf-set than $N$.

If $N$ contains a reticulate cherry, then perform a reticulate 
cherry reduction on $N$ to obtain a tree-child network $M$ (by Theorem~\ref{t:cherries}(ii)). 
Let $r$ be the reticulation arc which is removed in the reduction, and $r' \neq r$
be the reticulation arc that is incident with $r$.
Note that we must have either $r \in R \cap R'$ or
$r' \in R \cap  R'$, or else it is straight-forward to check that $S_R$ is
not isomorphic to $S_{R'}$, a contradiction. 
So, suppose $p$ is equal to $r$ or $r'$ and $p \in R \cap R'$. 
Then as $R\neq R'$ and $|R|=|R'|$, $R-\{p\}$ and $R'-\{p\}$ must both
be non-empty, and $R-\{p\} \neq R'-\{p\}$.
Moreover, we can consider $R-\{p\}$ and $R'-\{p\}$ as being 
contained in the set of reticulation arcs in $M$. But the reduced versions 
of $S_R$ and $S_{R'}$ in $M$ must be isomorphic, which 
contradicts the choice of $N$, since $M$ has less arcs than $N$.
\end{proof}

We now show that if two tree-child networks have the same set of 
embedded spanning trees, then they must be isomorphic. We begin
with a useful observation:

\begin{lem}\label{helper}
	Suppose that $N$ and $N'$ are tree-child networks on $X$. If, 
	for $x,y \in X$, either of the following hold:
	\begin{itemize}
		\item [(i)] $N$ contains a cherry $\{x,y\}$ and $N'$ does not; or
		\item [(ii)] $N$ contains a reticulate cherry $\{x,y\}$ with $y$ the 
	 leaf below the reticulation and $N'$ does not,
	\end{itemize}
	then $S(N) \neq S(N')$.
\end{lem}
\begin{proof}
(i) Suppose $N$ contains a cherry $\{x,y\}$ and $N'$ does not. Then
in the underlying graph for $N$ there is a path of length 2 between $x$ and $y$,
whereas in the underlying graph for $N'$ there is no such path. It easily
follows that $S(N) \neq S(N')$.

(ii) Suppose $N$ contains a reticulate cherry $\{x,y\}$ with $y$ the reticulation leaf,
and $N'$ does not, but that $S(N) = S(N')$. 
Note that every tree in $S(N)$ contains either 
\begin{enumerate}
	\item [(a)] a path $x,u,v,y$ of length 3 between
$x$ and $y$ with $v$ degree 2, or 
	\item [(b)] two paths of length 2 of the form $x,u',v'$ where $u'$ is
a degree 2 vertex and $y, u'', v''$ where $u''$ has degree 2, and no path 
between $x$ and $y$ of length less than 4,
\end{enumerate}
(see Figure~\ref{f:retic.cherry}).
Moreover, there
exists at least one tree in $S(N)$ which contains (a) and at least one that contains (b).
\begin{figure}[ht]
\includegraphics[]{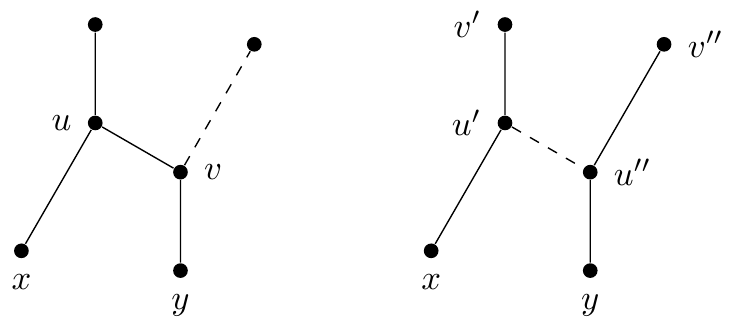}
\caption{}\label{f:retic.cherry}
\end{figure}

Since $S(N) = S(N')$, and since they are non-empty, there must be some tree in $S(N')$ which has a path
of length 3 between $x$ and $y$. Let $x,a,b,y$ be this path. Then
in the network $N'$ it is straight-forward to check that 
we must have one of the following possible cases: (1) $a$ is a tree vertex
and $b$ is a reticulation, (2) $b$ is a tree vertex and $a$ is a reticulation, or (3)
$a$ and $b$ are both tree vertices, or (4) one of $a,b$ is a tree vertex and the 
other the root vertex.

First, note that (1) is not possible as then 
$N'$ contains a reticulate cherry $\{x,y\}$ with $y$ the reticulation leaf. 
In case (2) it follows that  $N'$ contains a reticulate cherry $\{x,y\}$ with $x$ the reticulation leaf. 
But then there is no tree in $S(N')$ which contains structure (a), which contradicts $S(N)=S(N')$.
In cases (3) and (4), there is no tree in $S(N')$ which contains structure (b), again a contradiction.
\end{proof}

We are now able to prove the main result of this section, namely that sets of embedded spanning trees characterise tree-child networks.

\begin{thm}\label{isomorphic}
	Suppose that $N$ and $N'$ are tree-child networks on $X$. 
	Then $S(N)=S(N')$ if and only if $N$ is isomorphic to $N'$.	
\end{thm}

\begin{proof}
The reverse direction is immediate: networks that are isomorphic will have the same set of embedded spanning trees.  It remains to show that if $S(N)=S(N')$ then $N$ and $N'$ are isomorphic.

For the purposes of obtaining a contradiction, 
suppose that there exists a pair $N,N'$ of non-isomorphic 
tree-child networks on some set $X$, with $S(N) = S(N')$. 
It is straight-forward to check that $|X|>2$ using the observation made after Theorem~\ref{t:cherries} concerning
tree-child networks with two leaves. 
Take $|X|$ minimal for which there exists such a pair, and 
out of all these pairs on $X$, take a pair which minimizes
$\min\{|A(N)|,|A(N')|\}$ (without loss of generality suppose that this minimum is 
obtained for $N$). 

Consider the chosen minimal pair $N,N'$. Since $N$ is tree-child, 
it must contain either a cherry or a reticulated cherry, by Theorem~\ref{t:cherries}.
Moreover, it follows by Lemma~\ref{helper} that if $N$ contains a cherry $\{x,y\}$ 
then so must $N'$ (otherwise $S(N) \neq S(N')$), 
and that if $N$ contains a reticulate cherry $\{x,y\}$ with $y$ the reticulation leaf,
then so must $N'$.

Now, note that if $N$ and $N'$ are not isomorphic and both have a cherry $\{x,y\}$
(respectively both have a reticulated cherry $\{x,y\}$ with $y$ the reticulation leaf), 
then the tree-child networks $M$ and $M'$ obtained by performing a cherry reduction
on $\{x,y\}$ for $N$ and $N'$ (respectively a reticulated cherry reduction on $\{x,y\}$
for $N$ and $N'$) are not isomorphic. To see this, note that 
if $M$ and $M'$ are isomorphic, then we can extend the 
isomorphism to $N$ and $N'$.

Putting this together, if $N$ and $N'$ both contain a cherry $\{x,y\}$,
then perform a cherry reduction on both, to obtain two necessarily non-isomorphic
tree-child networks $M$ and $M'$ with $S(M)=S(M')$ (the last equality
follows from $S(N)=S(N')$). But this contradicts the choice of $N$ and 
$N'$ since the leaf-sets of $M$ and $M'$ are the same and this leaf-set 
is smaller than $X$.
And, if $N$ and $N'$ both contain 
a reticulate cherry $\{x,y\}$ with $y$ the leaf below the reticulation, then
perform a reticulate cherry reduction on both, to obtain two necessarily non-isomorphic
tree-child networks $M$ and $M'$ with $S(M)=S(M')$. But this 
again  contradicts the choice of $N$ and 
$N'$ since $M$ has a smaller number of arcs than $N$.
\end{proof}

\section{The tree model}\label{s:models}

\subsection{Evolution along a tree}\label{s:evol.on.tree}
We will consider the model $M=M(\mu)$, $\mu \in [0,1]$, of evolution of characters along branches of a
rooted tree with leaf-set $X$ as described by~\cite[p49]{thatte2013reconstructing}.   

A \emph{character} is a map from $X$ into an alphabet $\Sigma$, 
which for simplicity one can assume to be the set of nucleotides $\{A,C,G,T\}$.  An 
alignment is an $L$-tuple of characters on $X$, or a map $X\to \Sigma^L$.  If 
characters are considered as column vectors indexed by the leaf-set $X$, an 
alignment is an $|X|\times L$ array.  The columns of this array are called \emph{sites}.  
Thus an alignment is an array with rows labelled by elements of $X$ and columns 
labelled by the sites in the sequence, whose content at each site is the character 
value at the site.  Alignments can 
be considered elements of $(\Sigma^L)^{|X|}$, which we abbreviate $\Sigma^{XL}$, following~\cite{thatte2013reconstructing}.

Evolution of states on a tree under the model $M=M(\mu)$ requires setting an 
initial state for the root, and a rule for assigning states on vertices given the 
state of their parent vertex.  The root is assigned a state from $\Sigma$ uniformly 
at random with probability $\frac{1}{|\Sigma|}$. 
 Along an edge $(v,w)$, if $v$ is in state $x$, then there is probability $\mu$ that $w$ 
has state $y\in \Sigma\setminus \{x\}$.  Thus the probability of a change on a given edge 
is $(|\Sigma|-1)\mu$, and the probability of no change is $1-(|\Sigma|-1)\mu$.

As explained in~\cite[p.50]{thatte2013reconstructing}, model $M$ on a rooted $X$-tree
is equivalent to a similarly formulated model on the (undirected) tree $T$  that underlies it 
(that is, the (unrooted) tree with the same vertex set, and directions on all arcs ignored).
More specifically, suppose a root vertex is chosen arbitrarily in $T$, 
a letter from $\Sigma$ is assigned to it uniformly at random, and 
the state is then evolved along the edges away from the root.
Then, since the mutation model $M$ is reversible, the same distribution 
on the site patterns is observed on $X$ in the 
tree $T$ as in the rooted tree for a given $\mu$ (independent of the choice of root). 
In consequence, if we try to construct the rooted tree from the character distribution 
on its leaves, we can at best construct the underlying tree. Hence,  
in what follows we will  not differentiate between 
a rooted $X$-tree and the tree that underlies it 
when referring to model $M$.

\subsection{Identifiability for $X$-trees}\label{s:thatte.RPNs}

Let $p_i=p_i(T,\mu)$ be the probability of observing the character $C_i$ given 
the tree $T$ and the mutation model $M(\mu)$, that is, 
\[p_i=Pr\{C_i\mid T,M(\mu)\}.
\]
Then let $p(T,\mu)$ be the vector of probabilities of \emph{all} possible 
characters in an alignment, so that 
\[p(T,\mu)=(p_1,\dots,p_n),\] 
with $n:=|\Sigma|^{|X|}$ the 
number of possible characters.
This represents the theoretical distribution of character values predicted from the model.

In~\cite{thatte2013reconstructing}, a key identifiability result concerning collections of $X$-trees 
is presented which we now recall.
Given an alignment $A\in\Sigma^{XL}$, let $f(A):=(f_1,f_2,\dots,f_n)$, where $f_i$ is the 
proportion of columns of $A$ of type $C_i$ (the relative frequency of the character $C_i$ in $A$).
This is thus the \emph{observed} distribution of character values in the alignment.
In addition, let $\rho(s,r)$ denote the ball of radius $r$ around the point $s$ in $\R^n$, 
with distance defined by the $L_1$ (``taxicab'') metric.

Fix $\U$ to be a finite set of $X$-trees, and $r_0$ to be  half the smallest $L_1$ distance 
between frequencies 
predicted on distinct trees, that is,
\[r_0=\frac{1}{2}\min\left\{d(p(T,\mu),p(T',\mu))\mid T,T'\in \U, T\neq T'\right\}. 
\]
For $T$ an $X$-tree in $\U$, define 
\begin{equation}\label{eq:AT}
\A_T=\{A\in\Sigma^{XL}\mid f(A)\in \rho(p(T,\mu),r_0)\}.
\end{equation}  
The set $\A_T$, which depends on $r_0$, is the set of alignments for which 
the distribution of frequencies of characters is close to (within $r_0$ of) that of the 
theoretical prediction of evolution on tree $T$.

Now, for each $T \in \U$, let 
$$\epsilon_T= 1-Pr\{\A_T\mid T, M(\mu)\},$$ 
and put $\epsilon_{max}= \max_{T \in \U} \{\epsilon_T\}$. 
The probability $Pr\{\A_T\mid T, M(\mu)\}$ is the probability of observing 
alignments in the set $\A_T$, given the tree $T$ and model $M(\mu)$.  
That is, the probability of observing an alignment containing characters of 
frequencies close to those predicted theoretically.  

The following theorem shows that for sufficiently low mutation probability $\mu$, there 
is an alignment length $L$ that makes the probability of observing a member of $\A_T$ from 
the process on $T$ higher than $1-\epsilon_{max}$ and, at the same time, 
the probability of observing an element of $\A_T$ 
for the process on a tree $T'$ that is \emph{not isomorphic} to $T$ is smaller than $\epsilon_{max}$.
For its proof see \cite[page 59]{thatte2013reconstructing}.

\begin{thm}\label{t:XLT.alignment.bounded}
Let $\U$ be a finite set of $X$-trees, let $T\in\U$, and let $\mu \in (0,\frac{1}{|\Sigma|})$.   
Then there is an $L=L(T)$ such that 
\[1-Pr\{\A_T\mid T, M(\mu)\}<\epsilon_{max},\]
and
\[Pr\{\A_T\mid T', M(\mu)\}<\epsilon_{max}
\]
for all $T'\in \U\setminus\{T\}$.  
\end{thm}

\section{Evolution along a network}

\subsection{Description of the model}\label{s:Markov.on.network}

In the previous section, we described the evolutionary model $M$ for evolution along a tree; 
we now extend this model to networks, adapting the model for pedigrees in \cite{thatte2013reconstructing}. Our model will be defined
for networks $N$ such that every tree in $S(N)$ is an $X$-tree and $|S(N)|=2^{\omega}$, for 
$\omega=|R(N)|$,  holds. 

We first  define a Markov process on the set of rooted trees
$T(N)= \{T_R \,:\, R \in \mathcal{R}_N\}$. 
Given an element $T_R \in T(N)$, for each vertex $w\in R(N)$ we assign a 
fixed probability $p$ to make a change of vertex $w$'s parent to give another tree in $T(N)$. 
This describes a Markov chain on $T(N)$: the initial state given by taking a random choice of parent 
for each reticulation vertex in $N$ (probability $0.5$ assigned to each),
and the probability of being in any particular state (a tree in $T(N)$), at any point 
in the process, is uniform and equal to $\frac{1}{2^\omega}$.  

The Markov process that moves between trees in $T(N)$, together with the 
evolutionary model $M(\mu)$ for each tree now defines a network model under which 
characters evolve, which we denote $RM(\mu,p)$. 
Note that in this model, it is straight-forward to show using 
a similar argument to that used in the proof of 
\cite[Proposition 2]{thatte2013reconstructing}, that the probability of observing a 
character $C$ at the $k$th site of an alignment under $RM(\mu,p)$ is just 
the probability
of observing a given tree (which is $\frac{1}{2^\omega}$ since we are assuming $|S(N)|=2^{\omega}$), 
times the probability of 
observing the character on that tree (which, for $T \in S(N)$, is $Pr\{C\mid T,M(\mu)\}$
since we are assuming every $T\in S(N)$ is an $X$-tree), summed over all 
possible trees. That is,
\[
Pr\{C\mid N,RM(\mu,p)\}=\frac{1}{2^\omega}\sum_{T\in S(N)}Pr\{C\mid T,M(\mu)\}.
\]
Note that in particular, that this expression does not depend on $k$.

\subsection{Alignments arising from a network}\label{s:main.result}

This section adapts the approach to pedigrees used in~\cite{thatte2013reconstructing}, in order to 
derive similar results for networks.

The Markov process described in Section~\ref{s:Markov.on.network} moves around the set of 
$X$-trees $T(N)$ displayed by the network $N$.  By considering a sequence of 
characters generated on trees in this Markov chain, we are able to 
generate an alignment from $N$.  Such an alignment will be 
partitioned into a set of blocks, each of which arose from a particular tree.  
The following lemma describes how the probability 
of observing an alignment, given a rooted binary phylogenetic network $N$ 
and the model, can be computed.  It sums over cases according 
to the number of trees in the partition.  

A given sequence of trees $\T=(T_1,\dots, T_k)$ obtained from the Markov process has a 
sequence of transitions, each transition involving a certain number of changes of parent at 
reticulation vertices (this number of changes will be $\ge 1$, since adjacent trees 
in the sequence are distinct).  The total number $r(\T)$ of changes in the sequence of trees is given by
\[r(\T)=\sum_{i=1}^{k-1} \frac{1}{2}|E(T_{i+1})\triangle E(T_i)|,
\]
where $E(T)$ is the number of ``reticulation edges'' in $T$, namely edges that correspond to 
reticulation arcs in $N$, and $\triangle$ denotes the symmetric difference.  
Likewise, the total number of reticulation edges that are unchanged in transitions in 
the sequence, $s(\T)$, is given by \[s(\T)=\sum_{i=1}^{k-1} |E(T_{i+1})\cap E(T_i)|.\]

A \emph{composition} $\lambda$ of an integer $n$ is a sequence of positive integers that add to $n$, and is denoted $\lambda\vDash n$. 

The following Lemma~\ref{l:thatte.lem6} is a direct analogue of~\cite[Lemma 6]{thatte2013reconstructing}. 
Schematically, it computes the probability of an alignment by going from 
the network $N$ to the sequence of trees $\T$ (via the Markov 
process changing reticulation arcs), and from the sequence of trees to the alignment $A$.

\begin{lem}\label{l:thatte.lem6}
The probability of observing an alignment $A$ of length $L$ via the model $RM(\mu,p)$ on a network $N$ for which $|S(N)|=2^\omega$, is given by
\begin{align*}
& Pr\{A\mid N,RM(\mu,p)\}=\\
& \sum_{k=1}^L\left(\sum_{\substack{\mathbf T=(T_1,\dots,T_k)\\ T_i\neq T_{i+1}}}\left(\frac{p^{r(\T)}(1-p)^{s(\T)+\omega(L-k)}}{2^\omega}\sum_{(\ell_1,\dots,\ell_k)\vDash L}\left(\prod_{i=1}^kPr\left(A[L_{i-1}+1,L_i]\mid T_i,M(\mu)\right)
\right)
\right)
\right).
\end{align*}
\end{lem}

\begin{proof}
The alignment could be observed under any sequence of trees $\T$, and so 
we first break the problem into cases according to the length $k$ of this sequence, 
which can only be between 1 and $L$.  For each length of sequence $k$, we then 
sum over all possible sequences $\T$.  

The probability of observing the alignment on a particular sequence of 
trees $\T$ depends on the probability of observing the sequence $\T$, and the 
probability of the alignment given the particular trees in the sequence.

The probability of observing the sequence $\T$ is the probability of first 
observing the initial tree, $\frac{1}{2^\omega}$, times the probability of 
observing the numbers of recombinations and non-recombinations along the 
sequence, namely $p^{r(\T)}(1-p)^{s(\T)+\omega(L-k)}$.   

Finally, the probability of observing the alignment given the sequence of trees $\T$ \
depends on the lengths of the sub-alignments of $A$ that evolved on each tree (under $M(\mu)$).  
The possible lengths of the sub-alignments are given by the compositions of $L$ into $k$ parts, 
one for each tree in $\T$.  For a composition $(\ell_1,\dots,\ell_k)\vDash L$, set $L_i=L_{i-1}+\ell_i$, 
with $L_0=0$, to give the recombination sites (so that sites $L_{i-1}+1$ to $L_i$ evolved on tree $T_i$).  
Denote the sub-alignment of $A$ restricted to these sites by $A[L_{i-1}+1,L_i]$.  The probability 
of observing that sub-alignment on $T_i$ is then $Pr\left(A[L_{i-1}+1,L_i]\mid T_i,M\right)$, 
and the probability of observing the whole of $A$ given that sequence of trees $\T=(T_1,\dots,T_k)$ 
and that composition $(\ell_1,\dots,\ell_k)$ is the product of these over $i$ from 1 to $k$.  
\end{proof}

Lemma~\ref{l:thatte.lem6} shows how to compute a probability for each 
alignment $A$ of length $L$, given the network $N$ and model $RM(\mu,p)$, 
and so defines an alignment distribution $D_N=D_N(L,RM(\mu,p))$ which is given by the map
\[D_N:\Sigma^{XL}\to [0,1]\]
where
\[A\mapsto Pr\{A\mid N,RM(\mu,p)\}.\]

Two phylogenetic networks $N,N'$ in a class $\mathcal C$ are said to 
be \emph{distinguished} from one another under model $RM(\mu,p)$ 
if $D(\Sigma^{XL} \,|\, N, RM(\mu,p)) \neq D(\Sigma^{XL} \,|\, N', RM(\mu,p))$ for some $L$. 
That is, if there is some alignment $\mathcal A \subseteq \Sigma^{XL}$ such that
$Pr\{\mathcal A \,|\,  N, RM(\mu,p)\}\neq Pr\{\mathcal A \,|\,  N', RM(\mu,p)\}$.

\section{Main results}\label{s:main.results}

We now state and prove our main theorem, an 
analogue of~\cite[Theorem 2]{thatte2013reconstructing}.
This theorem states that for sufficiently large alignments, and 
a probability $p$, if the distributions of characters from two 
networks are the same, then any embedded spanning tree of one 
is also an embedded spanning tree of the other.  This will then imply
that the sets of embedded spanning trees for the two networks are the 
same (Corollary~\ref{first}).
	
\begin{thm}\label{t:main.result}
Suppose $N$ and $N'$ are phylogenetic networks on $X$, both with $\omega$ reticulations,
such that every tree in $S(N)$ and $S(N')$ is an $X$-tree and 
$|S(N)|=|S(N')|=2^{\omega}$. Let $\mu \in (0,\frac{1}{|\Sigma|})$.
If $T \in S(N)$, then there exists $L = L(T)\ge 1$ and $p=p(T) \in(0,1)$ such that if 
\[D(\Sigma^{XL}\mid N,RM(\mu,p))=D(\Sigma^{XL}\mid N',RM(\mu,p))\] 
then $T\in S(N')$. 
\end{thm}

\begin{proof}
The proof of the theorem is by contradiction, and because it is a 
complicated statement we first briefly review the logical structure, 
which is as follows: 
\begin{center}
if ``A'', then there are $L$ and $p$ such that ``B'' implies ``C''.  
\end{center}
Here ``A'' is $T\in S(N)$, ``B'' is $D(N)=D(N')$, and ``C'' 
is $T\in S(N')$, where $D(N)$ is short for $D(\Sigma^{XL}\mid N,RM(\mu,p))$.

To argue by contradiction, we assume the negation of ``there 
are $L,p$ such that B implies C'', that is, we assume 
``for all $L,p$, we have B and not C''.  In other words, 
we assume that $T\in S(N)$, and for all $L\ge 1$ and $p\in(0,1)$ 
we have $D(N)=D(N')$ but $T\not\in S(N')$.

We will show that for
some choice of $L$ (depending on $T$) and $p$ (depending on $L$ and therefore on $T$), we 
obtain a contradiction.

If the distributions are the same, then by definition the probabilities are the same for each alignment $A$, or set of alignments $\A\subseteq\Sigma^{XL}$.  That is,
\[Pr\{\A\mid N,RM(\mu,p)\}=Pr\{\A\mid N',RM(\mu,p)\},
\]
for each $\A\subseteq\Sigma^{XL}$.  
These probabilities are decomposed in Lemma~\ref{l:thatte.lem6} for each network, 
based on the number of trees in the sequence that generates the alignment.
We now break this decomposition into components 
according to whether there is a single tree in the sequence, so that $\T=(T_1)$, 
or whether there is more than one tree in the sequence.  Furthermore, if there is one tree in the sequence $\T=(T_1)$, we consider two cases: whether $T_1=T$ or not.
Thus, we write
this decomposition
\[Pr\{\A\mid N,RM(\mu,p)\}=P_{0,T}(N)+P_{0,\bar T}(N)+P_{>0}(N),
\]
indexing by the number of recombinations (0 or more):  $P_{0,T}(N)$ gives the component for $\T=(T)$; $P_{0,\bar T}(N)$ gives the component for $\T=(T_1)$ but $T_1\neq T$; and $P_{>0}(N)$ gives the component for all sequences $\T$ consisting of more than one tree.

A similar decomposition may be written for the network $N'$.  In the rest of the proof, we find expressions for these components for each of $N$ and $N'$, and use Theorem~\ref{t:XLT.alignment.bounded} to obtain upper and lower bounds for them, eventually choosing a value of $p$ that forces a contradiction.

The first case, that $\T=(T)$, gives contribution
\begin{equation}\label{e:P_0T}
P_{0,T}(N)=\frac{(1-p)^{\omega(L-1)}}{2^\omega} Pr\{\A\mid T,M(\mu)\},
\end{equation}
obtained by putting $k=1$ and $\T=(T)$ in the statement of Lemma~\ref{l:thatte.lem6}.  
Note that the coefficient here is the probability that $T$ is chosen as the first tree in the sequence (namely $\frac{1}{2^\omega}$, since $|S(N)|=2^\omega$) and subsequently no further trees are added through the Markov process ($(1-p)^{\omega(L-1)}$).  It follows that if $T$ is not displayed by the network $N'$ (as we have assumed), this term is zero: 
\begin{equation}\label{e:P_0T.N'}
P_{0,T}(N')=0.
\end{equation}

Likewise, the case of the sequence containing a single tree $T'\neq T$ is obtained by putting $k=1$ and summing over $\T=(T')$ for $T'\neq T$:
\begin{equation}\label{e:P_barT}
P_{0,\bar T}(N)=\sum_{T'\neq T}\frac{(1-p)^{\omega(L-1)}}{2^\omega} Pr\{\A\mid T',M(\mu)\}.
\end{equation}

The component for the remaining cases, in which the sequence $\T$ has more than one tree, is given by 
\begin{align}\label{e:P_>0}
&P_{>0}(N)=\\
&\sum_{k>1}^L\left(\sum_{\substack{\mathbf T=(T_1,\dots,T_k)\\ T_i\neq T_{i+1}}}\left(\frac{p^{r(\T)}(1-p)^{s(\T)+\omega(L-k)}}{2^\omega}\sum_{(\ell_1,\dots,\ell_k)\vDash L}\left(\prod_{i=1}^kPr\left(\A[L_{i-1}+1,L_i]\mid T_i,M(\mu)\right)
\right)
\right)
\right).\notag
\end{align}

We now use Theorem~\ref{t:XLT.alignment.bounded} to obtain bounds 
for each of these probabilities on $N$ and $N'$, for the particular set of alignments $\A_T$ 
whose character distribution is close to that predicted on $T$ (see Eq~\eqref{eq:AT}).  
We find that the first, $P_{0,T}(N)$, can be bounded from below, and the others bounded above, 
for suitable choice of $L$ (depending on $T$).

Let $\U=S(N)\cup S(N')$.   For $L$ sufficiently large and $\epsilon_{max}=\max_{T\in\U}\{\epsilon_T\}$ as defined in Section~\ref{s:thatte.RPNs}, 
\begin{align*}
P_{0,T}(N)&=\frac{1}{2^\omega}{(1-p)^{\omega(L-1)}} Pr\{\A_T\mid T,M(\mu)\}&& \text{by Eq.~\eqref{e:P_0T}}\\
&>
\frac{1}{2^\omega}{(1-p)^{\omega(L-1)}} (1-\epsilon_{max}) && \text{by Theorem~\ref{t:XLT.alignment.bounded}.}
\end{align*}
We have already noted in Eq.~\eqref{e:P_0T.N'} that the corresponding term for $N'$ is zero: $P_{0,T}(N')=0$.
For the second component, we have 
\begin{align*}
P_{0,\bar T}(N)&=\sum_{T'\neq T}\frac{1}{2^\omega}{(1-p)^{\omega(L-1)}} Pr\{\A_T\mid T',M(\mu)\}&&\text{by Eq.~\eqref{e:P_barT}}\\
&< (1-p)^{\omega(L-1)} \epsilon_{max}&&  \text{by Theorem~\ref{t:XLT.alignment.bounded},}
\end{align*}
noting that there are $2^\omega-1$ trees other than $T$.  
The critical point here is that this inequality also holds for the network $N'$.  This holds firstly because the decomposition in Eq.~\eqref{e:P_barT} is independent of the network, and secondly, the same inequality given in Theorem~\ref{t:XLT.alignment.bounded} with respect to the set of alignments $\A_T$ holds for each of the trees $T'\neq T$, and $|S(N')|=2^\omega$.  That is,
\[
P_{0,\bar T}(N')<(1-p)^{\omega(L-1)} \epsilon_{\max}.
\]

Finally, since each choice of recombination event is an instance of a binomially distributed random variable, in which there are $\omega(L-1)$ possible instances of events, each with probability $p$ (and noting the probability of any alignment on one of these trees is less than 1), we have 
\begin{align*}
P_{>0}(N),\ P_{>0}(N')&\le \omega(L-1)p.
\end{align*}
This uses the fact that if $X\sim Bin(n,p)$, then $Pr\{X\ge k\}\le \binom{n}{k}p^k$ (for us, $k=1$).

Now the assumption of the Theorem statement, that the distributions of alignments are the same from each network, implies $Pr\{\A\mid N,RM(\mu,p)\}=Pr\{\A\mid N',RM(\mu,p)\}$ for each set of alignments $\A$, and in particular for $\A_T$.  Thus,
\[
P_{0,T}(N)+P_{0,\bar T}(N)+P_{>0}(N)=P_{0,T}(N')+P_{0,\bar T}(N')+P_{>0}(N'),
\]
and since $P_{0,T}(N')=0$ (with $T\not\in S(N')$), we have 
\begin{align}\label{e:P.main.ineq}
P_{0,T}(N)
&=(P_{0,\bar T}(N')+P_{>0}(N'))-(P_{0,\bar T}(N)+P_{>0}(N))\notag\\
&< P_{0,\bar T}(N')+P_{>0}(N')
\end{align}
since the term subtracted is strictly positive.

Recall the bounds we have established above:
\begin{align}\label{e:P.bounds.recalled}
P_{0,T}(N)&>\frac{(1-p)^{\omega(L-1)}}{2^\omega} (1-\epsilon_{\max})\\
P_{0,\bar T}(N')&< (1-p)^{\omega(L-1)} \epsilon_{\max}\notag\\
P_{>0}(N')&\le \omega(L-1)p.\notag
\end{align}
Taking logs of both sides of the inequalities in Eq.~\eqref{e:P.bounds.recalled}, we obtain
\begin{align}
\log (P_{0,T}(N))&>\omega(L-1)\log (1-p)-\omega\log 2 + \log (1-\epsilon_{\max})\label{e:log.P0T}\\
\intertext{and}
\log (P_{0,\bar T}(N')+P_{>0}(N'))&<\log (P_{0,\bar T}(N'))+\log (P_{>0}(N')) \notag\\
&< \omega(L-1)\log (1-p)+ \log \epsilon_{\max} + \log(\omega(L-1)p).\label{e:log.RHS}
\end{align}

Combining inequality~\eqref{e:P.main.ineq} with inequalities~\eqref{e:log.P0T} and~\eqref{e:log.RHS}, we have
\[
\omega(L-1)\log (1-p)-\omega\log 2 + \log (1-\epsilon_{\max})
<
\omega(L-1)\log (1-p)+ \log \epsilon_{\max} + \log(\omega(L-1)p),
\]
which simplifies to 
\begin{equation}\label{e:final.inequality}
\log (1-\epsilon_{\max})
<
\omega\log 2 + \log \epsilon_{\max} + \log(\omega(L-1)p).
\end{equation}

Of the terms in these expressions, $\omega$ (the number of reticulations) is fixed by $N$ and $N'$, $\epsilon_{\max}$ is fixed, and $L=L(T)$ is a fixed value dependent on $T$.  However, if $p$ is chosen to satisfy
\[p<\frac{1-\epsilon_{\max}}{2^\omega\epsilon_{\max}\omega(L-1)},
\]
(noting that this value allows a choice of $p$ between 0 and 1, as required by the theorem statement), then 
\[\log(\omega(L-1)p)<\log\left(\frac{1-\epsilon_{\max}}{2^\omega\epsilon_{\max}}\right)=\log(1-\epsilon_{\max})-\omega\log 2-\log\epsilon_{\max}.
\]
This contradicts Inequality~\eqref{e:final.inequality}, and therefore our initial assumption that $T\not\in S(N')$ must be false, proving the claim in the theorem.  Note also that this choice of $p$ depends on $L$, which we have chosen to satisfy Theorem~\ref{t:XLT.alignment.bounded}, and so is dependent on $T$.  Consequently $p=p(T)$ is a function of $T$, also, as claimed.
\end{proof}

\begin{cor}\label{first}
	Suppose $N$ and $N'$ are phylogenetic networks on $X$ with $\omega$ reticulation vertices, such that 
	every tree in $S(N)$ and $S(N')$ is an $X$-tree and $|S(N)|=|S(N')|=2^{\omega}$. 
	Then there exists $L\ge 1$ and $p\in(0,1)$ such that if 
	\[D(\Sigma^{XL}\mid N,RM(\mu,p))=D(\Sigma^{XL}\mid N',RM(\mu,p))\] 
	then $S(N) =  S(N')$. 
\end{cor}

\begin{proof}
By Theorem~\ref{t:main.result}, there is an $L_1\ge 1$ and a $p_1\in(0,1)$ such that if 
\[D(\Sigma^{XL_1}\mid N,RM(\mu,p_1))=D(\Sigma^{XL_1}\mid N',RM(\mu,p_1))\] 
then $S(N) \subseteq S(N')$ (take $L_1$ to be the maximum over all $L(T)$, $T \in S(N)$ 
and $p_1$ to be the minimum $p_T$ with $T \in S(N)$). 

Likewise, there is an $L_2\ge 1$ and $p_2\in(0,1)$ such that if 
\[D(\Sigma^{XL_2}\mid N,RM(\mu,p_2))=D(\Sigma^{XL_2}\mid N',RM(\mu,p_2))\] 
then $S(N') \subseteq S(N)$. 
The result therefore follows by taking $L= \max\{L_1,L_2\}$ and
$p = \min\{p_1,p_2\}$.
\end{proof}	

We now state the main result of the paper.
We say that networks in a class $\mathcal C$ of phylogenetic networks are 
\emph{identifiable} under model $RM(\mu,p)$ if all pairs of networks in $\mathcal C$ 
are distinguished from each other under model $RM(\mu,p)$.

\begin{cor}\label{second}
	The class of tree-child networks on $X$ for which 
	the root does not form an arc with a reticulation vertex in the network, and
	such that every network in the class has the same number of reticulation vertices, is
	identifiable under model $RM(\mu,p)$. 
\end{cor}
\begin{proof}
  We need to show that if $N$ and $N'$ are in the given class with $N$ not isomorphic to $N'$, then 
  $D(\Sigma^{XL} \,|\, N, RM(p,\mu)) \neq D(\Sigma^{XL} \,|\, N', RM(p,\mu))$ for some $L\ge 1$.

	Suppose that $N$ and $N'$ are networks in the given class.
	As in both $N,N'$ the root does not form an arc with a reticulation vertex in the network,
	by Lemma~\ref{special} it follows that every tree in $S(N)$ and $S(N')$ is an $X$-tree. Moreover, 
	by Theorem~\ref{tight} we have $|S(N)|=|S(N')|=2^{\omega}$, where $\omega$ is
	the number of reticulation vertices in both $N$ and $N'$. Hence, by Corollary~\ref{first}
	there exists some $L'$ such that if $D(\Sigma^{XL'} \,|\, N, RM(\mu,p)) = D(\Sigma^{XL'} \,|\, N', RM(\mu,p))$
	then $S(N)=S(N')$,  which also implies that $N$ and $N'$ 
	are isomorphic by Theorem~\ref{isomorphic}. 
	Hence if $N$ and $N'$ are \emph{not} isomorphic, then 
	$D(\Sigma^{XL'} \,|\, N, RM(\mu,p)) \neq D(\Sigma^{XL'} \,|\, N', RM(\mu,p))$ for $L'$, as required.  
\end{proof}

\section{Discussion}

In this paper, we have shown that we can identify a certain subclass of tree-child networks 
under the model $RM(\mu,p)$. It would be interesting to see if this 
could be extended to the class of all tree-child networks, or 
to other classes of networks. We 
note that out model was defined for networks all of whose embedded spanning trees are 
$X$-trees; for more general networks this may not be the case, but 
this can probably be adjusted for using techniques developed for pedigrees in 
\cite{thatte2013reconstructing} (although probably at the expense of requiring
more technical arguments). In another direction, it could be worth  investigating what happens
when the model $RM(\mu,p)$ is extended to allow independent probabilities 
at each recombination vertex (instead of setting them all equal to $p$). 

Many of the questions raised in \cite[Section 6]{thatte2013reconstructing} for pedigrees have 
natural analogues for networks. For example, Corollary~\ref{first}
tells us that if $N$ and $N'$ are phylogenetic networks on the same leaf-set
that satisfy certain conditions and induce the same distributions, then 
$S(N)=S(N')$. But is it possible to prove some type of converse for 
this statement? Moreover, in practice it could be 
computationally expensive to check the condition $S(N)=S(N')$, and 
so the question arises as whether or not there are
there are possibly more tractable combinatorial conditions for checking when two networks can be
distinguished relative to model $RM(\mu,p)$?

Finally, it would be interesting to see if model $RM(\mu,p)$ might 
provide useful  information in addition to purely combinatorial invariants for
identifying networks. For example, in \cite{huber2014much} it is shown that 
certain pairs of phylogenetic networks cannot be distinguished from 
one another even by comparing all of the possible subtrees and networks that they display. It
would be interesting to know if they can however be distinguished 
under model $RM(\mu,p)$. 

\section*{Acknowledgements}
The authors thank the Royal Society for its support for this collaboration via an International Exchanges award.


\end{document}